\documentclass[10pt]{article}
 \usepackage{fullpage}
	\usepackage{graphicx}         
	\usepackage{amsmath, amsthm, amssymb}
	\usepackage[font=small,labelfont=bf]{caption}
	\usepackage[newfloat=true,frozencache=true]{minted}
	\usepackage{subcaption}
	\newtheorem{definition}{Definition}
	\newtheorem{theorem}{Theorem}

	\title{Fast Snapshottable Concurrent Braun Heaps}
	\author{Thomas D. Dickerson}
\begin{document}
	\maketitle
	\begin{abstract}
		This paper proposes a new concurrent heap algorithm, based on a stateless shape property, which efficiently maintains balance during insert and removeMin operations implemented with hand-over-hand locking.
		It also provides a O(1) linearizable snapshot operation based on lazy copy-on-write semantics.
		Such snapshots can be used to provide consistent views of the heap during iteration, as well as to make speculative updates (which can later be dropped).
		
		The simplicity of the algorithm allows it to be easily proven correct, and the choice of shape property provides priority queue performance which is competitive with highly optimized skiplist implementations (and has stronger bounds on worst-case time complexity).
		
		A Scala reference implementation is provided.
	\end{abstract}
	\section{Introduction}
		Concurrent collections data structures are an important tool in modern software design, but they often provide a limited set of operations compared to their serial counterparts, due to the complexity of concurrent reasoning and implementation.
		The priority queue is a widely used data structure, that allows a user to efficiently query and remove its smallest element.
		These operations impose a weaker ordering requirement than other collections (e.g. an ordered map), and it is thus potentially more efficient to implement those operations.
		Serial priority queues are typically implemented as either a heap or a skiplist, and heaps have stronger guarantees of efficiency.
		Unfortunately, tree data structures have proven more difficult to efficiently parallelize, and so most concurrent priority queues are based on some form of skiplist.~\cite{shavit2000skiplist}
		
		This paper defines a new concurrent heap algorithm with a shape property that is better tailored to concurrent updates,
		and augments it with a fast snapshot operation based on a copy-on-write methodology.
		Snapshots can be used to provide read-only iterator semantics, but they can also be used to perform speculative updates, which are a key component in wrapping concurrent data structures with a transactional API or implementing other kinds of reversible atomic objects.~\cite{dickerson2017proust,antonopoulostheory}
		
		A Scala reference implementation of the algorithm is provided in the appendix.
	\subsection{Prior Work}
	\begin{figure*}
		\centering
		\begin{subfigure}[t]{0.3\textwidth}
			\includegraphics[width=\textwidth]{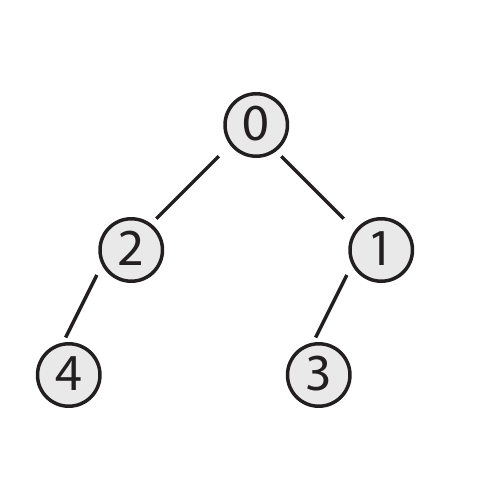}
			\caption{\label{fig:cbheapinit} The initial state.}
		\end{subfigure}
		\begin{subfigure}[t]{0.3\textwidth}
			\includegraphics[width=\textwidth]{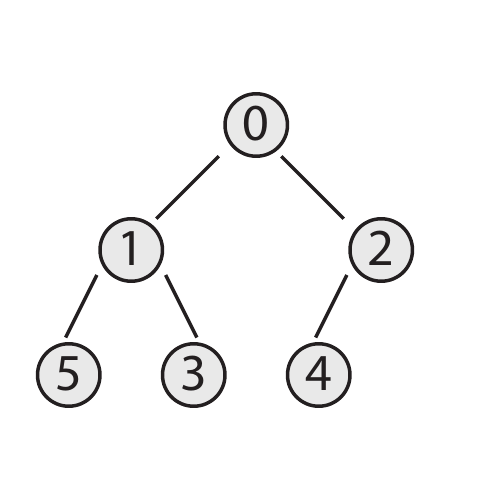}
			\caption{\label{fig:cbheapinsert} After \texttt{insert(5)}.}
		\end{subfigure}
		\begin{subfigure}[t]{0.3\textwidth}
			\includegraphics[width=\textwidth]{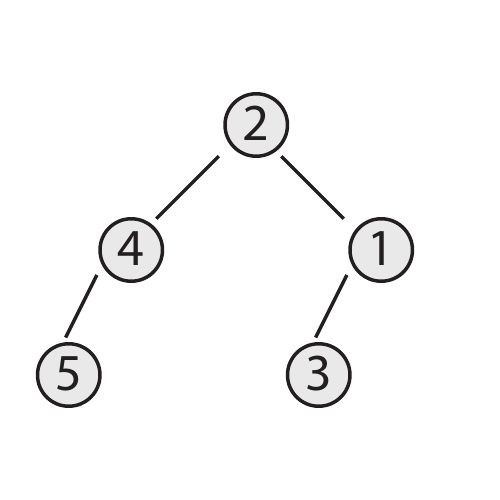}
			\caption{\label{fig:cbheapdelete} After \texttt{removeMin}.}
		\end{subfigure}
		\caption{A sequence of operations on a Braun heap, initialized with \texttt{(0 to 4).foreach(heap.insert(\_))}.}
	\end{figure*}
	This paper builds on several prior lines of work.
	
	Efficient tree data structures rely on a shape (or balance) property to ensure that the depth of the tree grows only logarithmically in the number of elements.
	The Braun property states that for every node, $n$, in a tree, the size of its left and right subtrees, $l$ and $r$, are constrained such that $|r| \le |l| \le (|r|+1)$.
	It was first proposed for efficient implementations of immutable trees in a pure-functional setting \cite{braun1983logarithmic,hoogerwoord1993logarithmic}, and while relatively unknown in the broader community\footnote{Unlike many algorithms and data structures, there is as of yet no Wikipedia article on the subject, which seems a reasonable measure of obscurity}, Braun trees have long been known to be useful in implementing heaps.~\cite{nipkow2016priority,okasaki1997three}
	An example of operations on a Braun heap is provided in Figures \ref{fig:cbheapinit}, \ref{fig:cbheapinsert}, and \ref{fig:cbheapdelete}.
	
	Collections data structures are of limited utility if their API only allows you to inspect or alter a single element at a time, yet bulk read operations are challenging in a concurrent context.
	Correctness of concurrent operations is often defined in terms of linearizability \cite{herlihy1990linearizability}, which ensures that all operations appear to happen instantaneously; however, a bulk read operation necessarily must capture the state of the entire data structure.
	Several recent works have leveraged a copy-on-write methodology to provide efficient linearizable snapshots for a number of data structures, including AVL trees \cite{bronson2010practical} and hash tries \cite{prokopec2012concurrent}. 
	Linearizable iterators are also an area of inquiry, recent work has implemented them for skip lists \cite{petrank2013lock}, which are of interest for implementing priority queues.
	
	Additionally, there are a number of other concurrent skiplist variants, which might be used as the basis for a priority queue.~\cite{herlihy2007simple,pugh1990concurrent,dleaCSLM}
	\subsection{Overview}
	It is unsurprising that importing insights from the purely-functional data structures community would be of benefit in designing new concurrent data structures, particularly those requiring a fast-snapshot operation.
	Without mutation, \textbf{every} reference to a data structure is intrinsically a snapshot (this property is known as persistence \cite{okasaki1996purely}), because every update operation returns a new data structure (potentially sharing references to sub-structures in the old one).
	Since allocation of new objects is a relatively expensive operation, it is beneficial to design data structures that do not require extensive modifications to maintain bookkeeping state or enforce algorithmic invariants.
	Moreover, every point where an algorithm on purely functional data structures would allocate a new object, a concurrent algorithm implementing the same update as a mutation will require synchronization, potentially resulting in a bottleneck.
	
	Furthermore, functional programming separates reasoning about the result of a computation from reasoning about its execution, allowing a program to be written correctly even if the result of a computation is delayed\footnote{In a single threaded setting, this is referred to as laziness. In a concurrent setting it is often referred to by the names ``future'' or ``promise'' \cite{friedman1976impact,baker1977incremental}.}.
	
	The Braun property in particular can be maintained by following two simple rules, without any bookkeeping.
	First, every insertion or deletion must swap the left and right subtrees.
	Second, insertions traverse towards the (new) left subtree, and deletions proceed towards the (new) right subtree.
	As an added benefit for concurrency, for every consecutive pair of updates, $u,v$, to a subtree rooted at a node $n$, the traversal paths of $u$ and $v$ will intersect only at $n$.
	Thus, the number of threads that can be working simultaneously doubles at each level of the tree.
	On the other hand, in a traditional binary heap, utilizing the fullness property, every insertion and deletion is in contention for the boundary position between filled and unfilled leaves, and, with high probability, every consecutive pair of updates will intersect for the majority of their path towards the root\footnote{Intuitively, they will intersect after 1 step with probability $\frac{1}{2}$, after 2 steps with probability $\frac{3}{4}$, etc.}.
	
	A sequential implementation of a Braun heap is provided in Listing \ref{lst:seqbheap}, consisting of two classes: heap nodes, containing a value and two child references, and the heap itself, containing a single reference to the root.
	Our concurrent version preserves the essence of this implementation (see Listings \ref{lst:parbheaphelpers},\ref{lst:parheapnode1},\ref{lst:parheapnode2},\ref{lst:parbraunheap} for details), while providing thread-safety and linearizable snapshots.
	With one exception, the effects of mutating operations only propagate leafward through the tree, so thread safety is achieved simply by following a hand-over-hand locking discipline.
	In the case of \texttt{removeMin}, the mutations are separated into two phases, the first to maintain the Braun property, the second to maintain the heap property.
	In the first phase, nodes must be able to reach upwards to modify references to themselves held by their immediate parents; however, at the beginning of the first phase, the mutex on the root is reentrantly acquired a second time to prevent observation while the heap property is violated, so thread safety is not violated by this minor breach of hand-over-hand discipline.
	
	The fast snapshot functionality is implemented as a reference-counting lazy copy-on-write (COW) scheme.
	Each node is augmented with a \texttt{snapCount} property, which tracks the number of snapshot references to that node (it is initialized to 0, since the copy in the ``original'' heap is not considered a snapshot).
	When a snapshot of a heap is requested, a new heap\footnote{Recall that the heap class serves only to hold a reference to the root.} is allocated with a reference to the same root node, and the \texttt{snapCount} of that node is incremented.
	Mutation of any node with \texttt{snapCount > 0} is forbidden, as this would lead to inconsistent state in the other heaps that reference it.
	Instead, a new node is allocated with the updated values, and the \texttt{snapCount}s of its left and right children are incremented, and the \texttt{snapCount} of the original node, for which mutation was requested, is decremented.
	This has the effect of lazily ``peeling'' the snapshot away from the original (i.e. only as-needed).
	
	In a language with deterministic destructors, it would be correct for a heap to decrement the \texttt{snapCount} of any snapshotted nodes when it exits scope; however, on the JVM \texttt{finalize} is often executed much later, and any unnecessary allocations resulting from an imprecise \texttt{snapCount} are likely to have already occurred.
	Thus, this implementation does not implement a \texttt{finalize} hook for that traversal.
	\begin{listing*}
	\begin{minted}{scala}
class HeapNode[E:Ordering](var v: E, var left:HeapNode[E] = null,
                                     var right:HeapNode[E] = null){
  def update(nV:E, nL:HeapNode[E], nR:HeapNode[E]):HeapNode[E] = {
    v = nV; left = nL; right = nr; this
  }
  def insertHelper(nV: E):HeapNode[E] = {
    val (smaller, larger) = if(nV < v) { (nV, v) } else { (v, nV) }
    update(smaller, if(right == null) { new HeapNode(larger) } 
                    else { right.insertHelper(larger) }, left)
  }
  def pullUpLeftHelper(selfP:HeapNode[E] => Unit, isRoot:Boolean = false):Option[E] = {
    if(left == null) {
      selfP(null)
      if(root == this) { None }
      else { Some(value) }
    } else {
      update(v, right, left)
      right.pullUpLeftHelper(right_= _)
    }
  }
  def pushDownHelper:Unit = {
    val pushNext = if(right == null && left != null) {
      val tmpLV = left.v
      if(tmpLV < v){ left.update(v); v = tmpLV }
      null
    } else if(right != null) {
      val (tmpV, tmpLV, tmpRV) = (v, left.v, right.v)
      (if(tmpLV < tmpV || tmpRV < tmpV) { v = tmpLV; left }
       else { v = tmpRV; right }).update(tmpV)
    } else { null }
    if(pushNext != null) pushNext.pushDownHelper
  }
}
class BraunHeap[E:Ordering] {
    var root:HeapNode[E] = null
    def getMin:Option[E] = Option(root).map(_.v)
    def insert(v:E):Unit = if(root == null){ root = new HeapNode(v) }
                             else { root.insertHelper(v) }
    def removeMin:Option[E] = getMin.map{ oldV =>
        root.pullUpLeftHelper(root_= _, true).foreach{ newV =>
          root.v = newV; root.pushDownHelper
        }
        oldV
    }
}
	\end{minted}
	\caption{\label{lst:seqbheap} Sequential Logic for a Braun Heap}
	\end{listing*}
	
	\section{Correctness}
	The correct execution of operations on the concurrent Braun Heap is subject to three invariants.
	The Braun property has already been defined, but we restate it here for convenience, in Definition \ref{def:braun}.
	\begin{definition}[\label{def:braun}Braun Tree]
	A binary tree, $T$, is Braun iff for every node, $n \in T$,  the size of its left and right subtrees, $n.l$ and $n.r$, are constrained such that $|n.r| \le |n.l| \le (|n.r|+1)$.
	\end{definition}
	The heap property is well known, but we use a non-standard formulation (Definition \ref{def:heap}) conducive to describing the behavior of threads traversing a tree.
	\begin{definition}[\label{def:heap}Heap]
	A tree, $T$, is a heap iff the values held by the nodes in every rootward path are nonincreasing.
	\end{definition}
	Snapshot isolation is well studied in the database literature, but our application is concerned with the evolution of data structures subject to linearizable operations, and not subjected to transactional semantics.~\cite{berenson1995critique,daudjee2006lazy}
	Nevertheless, we appropriate the terminology and provide a closely related\footnote{An optimistic STM wrapper built using the techniques of \cite{dickerson2017proust} would provide snapshot isolation to any transactions accessing the heap.} Definition \ref{def:snap}.
	\begin{definition}[\label{def:snap}Snapshot Isolation]
	Two heaps, $H_1, H_2$, subject to structural sharing, provide snapshot isolation, iff, without loss of generality, no modification of $H_1$ is reflected in the multiset of values held by $H_2$.
	\end{definition}
	
	If every operation provided by the heap individually preserves the Braun, Heap, and Snapshot Isolation properties, as defined above, and appears to take place at an atomic linearization point, then the algorithm is correct.
	We will explicitly consider the \texttt{insert}, \texttt{removeMin}, \texttt{getMin}, and \texttt{snapshot} operations as implemented in Listing \ref{lst:parbraunheap}, as well as the \texttt{pullUpLeftHelper} and \texttt{pushDownHelper} helper methods that \texttt{removeMin} is built on.
	For purposes of reasoning about atomicity through a hand-over-hand traversal, we will consider the heap's lock mediating access to the root to be the base case, rather than the locks owned by the root itself.
	
	\begin{theorem}[Correctness of getMin] 
	\texttt{getMin} happens atomically, and is Braun, Heap, and Snapshot Isolation preserving.
	\end{theorem}
	\begin{proof}
	 \texttt{getMin} is a read-only operation, so Braun, Heap, and Snapshot Isolation are preserved. 
	 \texttt{getMin} acquires a read-lock on the heap, and subsequently on its root node, thus it appears to happen atomically at the moment when the read lock is acquired on the heap, as no conflicting operations can occur once that lock is acquired.
	\end{proof}
	
	\begin{theorem}[Correctness of snapshot] 
	\texttt{snapshot} happens atomically, and is Braun, Heap, and Snapshot Isolation preserving.
	\end{theorem}
	\begin{proof}
	When a thread $t$ executes \texttt{snapshot} on a heap, $h$, rooted at node $n$, it increments the \texttt{snapCount} of $n$, and allocates a new heap, $h'$, with $h'.n = n$.
	
	 \texttt{snapshot} does not alter either a node's stored value, or its children, so Braun, Heap, and Snapshot Isolation are preserved.
	 \texttt{snapshot} acquires a write-lock on the heap, and subsequently on its root node, thus it appears to happen atomically at the moment when the write lock is acquired on the heap, as no conflicting operations can occur once that lock is acquired.
	\end{proof}
	
	\begin{theorem}[Correctness of insert] 
	\texttt{insert} happens atomically, and is Braun, Heap, and Snapshot Isolation preserving.
	\end{theorem}
	\begin{proof}
	 \texttt{insert} acquires a write-lock on the heap and observes hand-over-hand locking along its leafward path.
	 At step $i$, no previous conflicting operation can intrude, as the lock for node $n_i$ on the path is held, and no conflicting operation can be observed, as the lock for $n_{i+1}$ will be acquired before proceeding
	 \footnote{Note that this holds even when $n_i$ is structurally shared as the result of a snapshot operation.
	 If $n_i \in h$ and $n_i \in h'$, then competing threads performing inserts on $h$ and $h'$ must still acquire the same lock on $n_i$}.
	 Inductively, the entire traversal appears to happen atomically at the moment when the write lock is acquired on the heap.
	 
	 When a thread executing \texttt{insert(v)} arrives at a subtree, $t = \{t.v,t.l,t.r\}$, an \texttt{update} is performed such that,
	 $t' = \{\mathtt{min}(t.v,v),t.r.\mathtt{insert}(\mathtt{max}(t.v,v)) ,t.l\}$, and $t'$ is returned.
	 
	 If $t$ is a snapshot, then $t'$ is the result of allocating a new node (with snapshots of its children), otherwise $t$ is mutated directly, thus Snapshot Isolation is preserved.
	  
	 Assume that when the thread arrives at $t$, then $t$ is Braun.
	 Recall that Braun implies $|t.r| \le |t.l| \le |t.r|+1$.
	 Thus after \texttt{insert}, we have that $|t'.l| = |t.r|+1$ and $|t'.r| = |t.l|$.
	 Thus $|t'.r| \le |t'.l| \le |t'.r|+1$ and Braun holds for $t'$.
	 
	 Assume that when the thread arrives at $t$, then $t$ is a heap, implying $t.v \le t.r.v$.
	 Since $t'.v = \mathtt{min}(v, t.v)$ and $t'.l.v = \mathtt{min}(t.r.v,\mathtt{max}(v,t.v))$, then $t'.v \le t.v \le t'.l.v$.
	 Thus $t'$ is also a heap.
	\end{proof}
	
	\begin{theorem}[Correctness of pullUpLeftHelper ]
	If invoked by \texttt{removeMin},  \texttt{pullUpLeftHelper} happens atomically, and is Braun, Heap, and Snapshot Isolation preserving.
	\end{theorem}
	\begin{proof}
	\texttt{pullUpLeftHelper} is called with a write-lock on the root and observes hand-over-hand locking along its leafward path.
	The atomicity argument is largely as for \texttt{insert}; however, upon arriving at the target leaf, the reference to that leaf in its immediate parent is deleted.
	Despite this violation of hand-over-hand locking, \texttt{pullUpLeftHelper} maintains its atomicity if the invoking code retained a lock on the heap preventing any other threads from traversing down after it.
	
	When a thread executing \texttt{pullUpLeftHelper} arrives at a non-leaf subtree, $t = \{t.v,t.l,t.r\}$, an \texttt{update} is performed such that,
	$t' = \{t.v,t.r,t.l\}$, and then \texttt{pullUpLeftHelper} is invoked on $t'.r$.
	If $t'.r$ is a leaf, it is deleted, and $t'.r.v$ is returned.
	
	 If $t$ is a snapshot, then $t'$ is the result of allocating a new node (with snapshots of its children), otherwise $t$ is mutated directly, thus Snapshot Isolation is preserved.
	 
	 Assume that when the thread arrives at $t$, then $t$ is Braun.
	 Thus after \texttt{pullUpLeftHelper}, we have that $|t'.r| = |t.l|-1$ and $|t'.l| = |t.r|$.
	 Thus $|t'.r| \le |t'.l| \le |t'.r|+1$ and Braun holds for $t'$.
	 
	 Assume that when the thread arrives at $t$, then $t$ is a heap.
	 No elements are inserted below $t$, thus every rootward path through $t$ remains nonincreasing, and \texttt{pullUpLeftHelper} preserves heap.
	\end{proof}
	
	\begin{theorem}[Correctness of pushDownHelper ]
	If invoked by \texttt{removeMin},  \texttt{pushDownHelper} happens atomically, and is Braun and Snapshot Isolation preserving, and Heap restoring.
	\end{theorem}
	\begin{proof}
	\texttt{pushDownHelper} is called with a write-lock on the root and observes hand-over-hand locking along its leafward path.
	The atomicity argument is as for \texttt{insert}.
	
	When a thread executing \texttt{pushDownHelper} arrives at a non-leaf $t = \{t.v,t.l,t.r\}$, \texttt{update}s are performed such that,
	$t'.v = \mathtt{min}(t.v, t.l.v, t.r.v)$, $c'.v = \mathtt{max}(t.v,c.v)$, $t'.c = c'$, where $c$ selects whichever of $t.l$ and $t.r$ contains the smaller value.
	Then \texttt{pushDownHelper} is invoked on $t'.c$.
	
	 If $t$ is a snapshot, then $t'$ is the result of allocating a new node (with snapshots of its children), otherwise $t$ is mutated directly.
	 If $c$ is a snapshot, then $c'$ is the result of allocating a new node (with snapshots of its children), otherwise $c$ is mutated directly.
	 Thus Snapshot Isolation is preserved.
	 
	 Assume that when the thread arrives at $t$, then $t$ is Braun.
	 Since no elements are added or removed, $t'$ is also Braun.
	 
	 Assume that when the thread arrives at $t$, the subtrees rooted at its children are heaps.
	 For children $c,d \in t$, we say without loss of generality that $c.v \le d.v$.
	 After \texttt{pushDownHelper}, we have that $t'.v = \mathtt{min}(c.v,t.v)$, $c'.v = \mathtt{max}(c.v,t.v)$, $d'.v = d.v$.
	 Thus $t'.v \le c'.v$ and also $t'.v \le d'.v$, and inductively, \texttt{pushDownHelper} restores heap.
	\end{proof}
	
	\begin{theorem}[Correctness of removeMin]
	\texttt{removeMin} happens atomically, and is Braun, Heap, and Snapshot Isolation preserving.
	\end{theorem}
	\begin{proof}
	After acquiring the heap, \texttt{removeMin} reentrantly acquires a double write-lock on the root, $r.v$, and updates the root, s.t. $r'.v = r.\mathtt{pullUpLeftHelper}$.
	After this, \texttt{removeMin} executes $r'.\mathtt{pushDownHelper}$.
	
	Since \texttt{pullUpLeftHelper} ensures that $r'$ is not a snapshot, the direct assignment to $r'.v$ does not violate Snapshot Isolation.
	
	Since the lock on the root is held continuously between starting \texttt{pullUpLeftHelper} and starting \texttt{pushDownHelper}, \texttt{removeMin} is atomic.
	
	Both subroutines are Braun preserving, so too is their composition.
	
	Since \texttt{pullUpLeftHelper} is heap preserving,  $r'.l$ and $r'.r$ are heaps, \texttt{pushDownHelper} is heap restoring, and \texttt{removeMin} is heap preserving.
	\end{proof}

	\section{Performance}
	We benchmarked the concurrent Braun heap implementation of this paper against two other concurrent priority queue implementations.
	\texttt{PriorityBlockingQueue} is part of the Java standard library, and its methods use a simple mutex around a high-performance single-threaded binary heap implemented with an array.~\cite{dleaPBQ}
	It provides no native snapshot operation; however, while the documentation asserts that its iterator is ``weakly consistent'', the source of the implementation reveals that it uses a mutex for a copy of the backing array, which provides strong consistency.
	The provided constructor for efficient construction from sorted bulk data is incompatible with the \texttt{Iterator} interface, so we implemented \texttt{snapshot} with successive inserts from the iterator.
	
	\texttt{SkipListPriorityQueue} is built on a variant of the Java standard library's \texttt{ConcurrentSkipListMap} that has been augmented with the lock-free linearizable iterator of Petrank and Timnat.~\cite{dleaCSLM,petrank2013lock}
	Simple CAS loops built on \texttt{putIfAbsent}, \texttt{replace}, and the key/value variant of \texttt{remove} were used to transform the ordered map into an ordered multiset supporting \texttt{insert} and \texttt{removeMin}.
	Their iterator was augmented with a wrapper to properly repeat each entry in the multiset.
	As with \texttt{PriorityBlockingQueue}, no native snapshot was supported, and no compatible bulk data constructor was available, so we implemented it with successive inserts from the iterator.
	
	\subsection{Methodology}
	The experiments were run on an Amazon EC2 ``c4.8xlarge'' instance, providing 36 virtual CPUs and 60GB of RAM, using an Oracle Java 8 runtime.
	For each sequence of tests, the priority queue implementation being benchmarked was initialized with a large quantity of random entries,
	then a series of warmup executions for the VM's JIT compiler were performed and discarded. Finally, each sequence of tests were performed and the times recorded and averaged over a number of iterations.
	For experiments where the priority queues were initialized with $2^{20}$ entries, 20 warmup runs and 20 experimental runs were used.
	For experiments where the priority queues were initialized with $2^{25}$ entries, 1 warmup run and 5 experimental runs were used due to the substantially longer execution times involved in the iterator tests.
	
	In all, five tests were conducted, Insert, RemoveMin, Sum, Snap+Insert, and Snap-Only.
	For the Insert tests, 1344 randomly selected values were inserted into the priority queue, with each of $t$ threads responsible for $1344/t$ insertions\footnote{
	This number was selected because it is a common multiple of the thread counts used, to avoid rounding error}.
	Similarly, for the RemoveMin tests, each of $t$ threads was tasked with executing RemoveMin $1344/t$ times.
	
	For the Sum tasks, all threads were tasked with individually calculating the sum of every entry in the priority queue.
	
	For the Snap+Insert tasks, each of $t$ threads created its own snapshot and inserted $1344/t$ items, whereas for the Snap-Only tasks, each of $t$ threads merely had to instantiate a snapshot.
	
	\subsection{Results}
	\begin{figure*}
		\centering
		\includegraphics[width=0.48\textwidth]{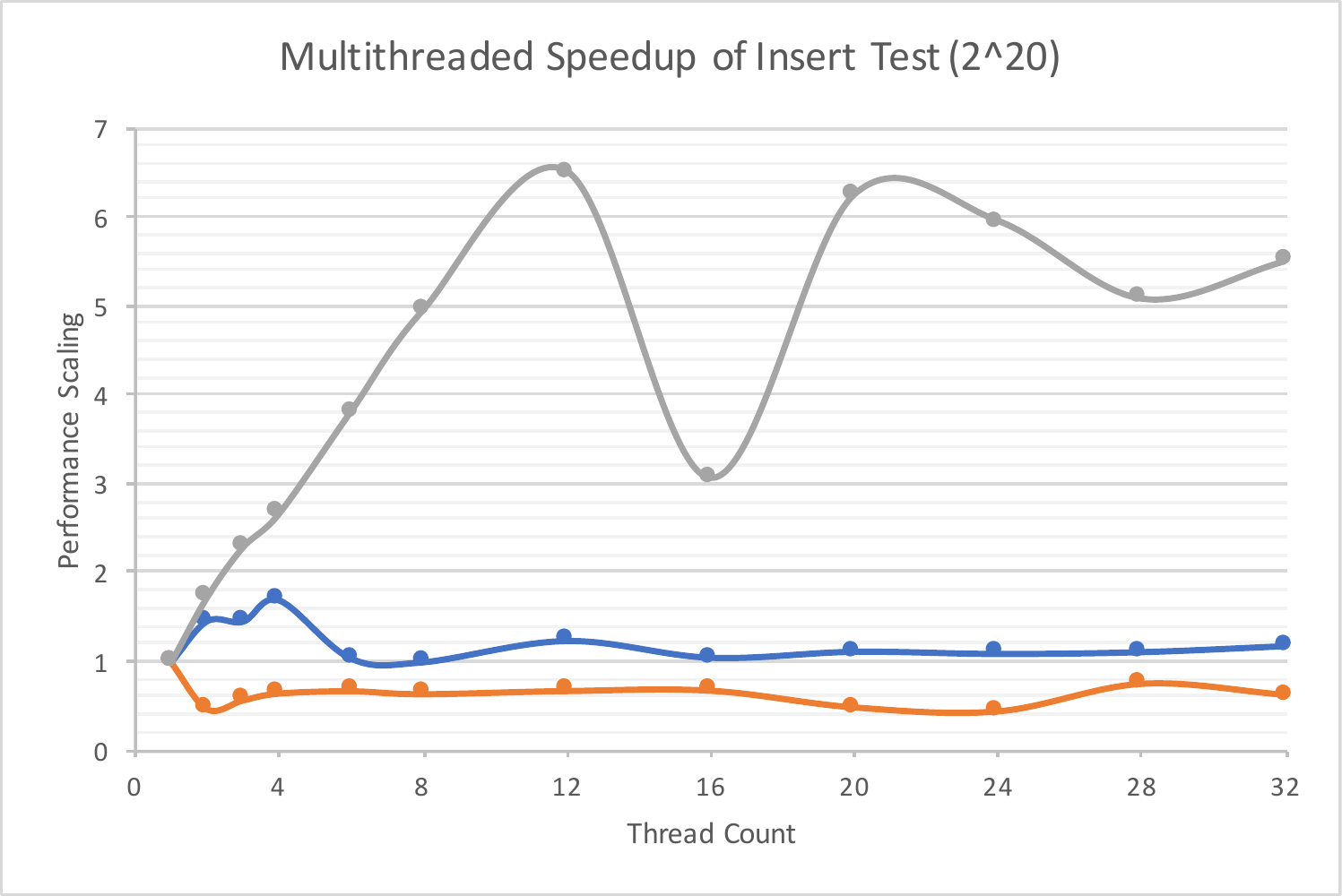}
		\includegraphics[width=0.48\textwidth]{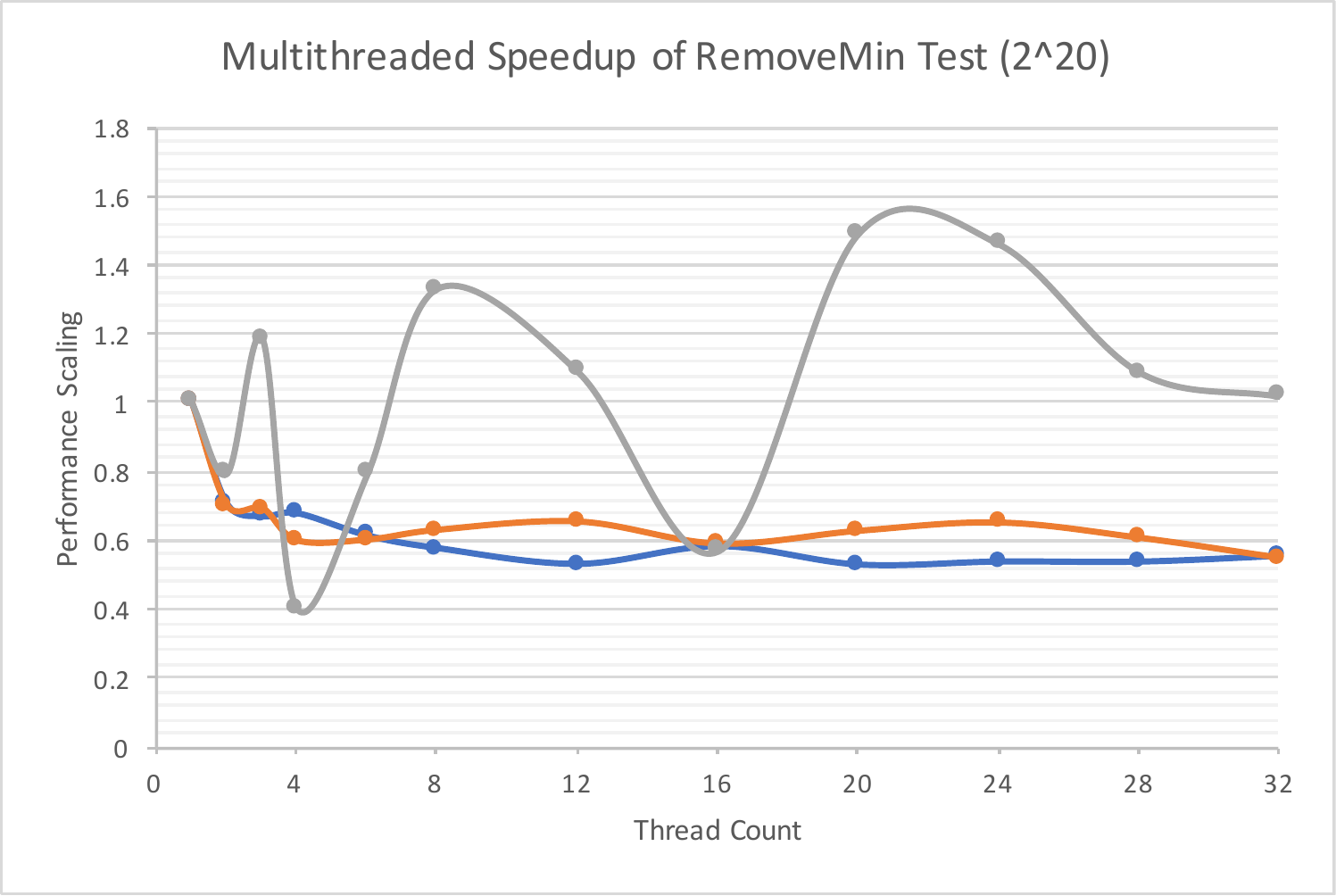}
		\includegraphics[width=0.48\textwidth]{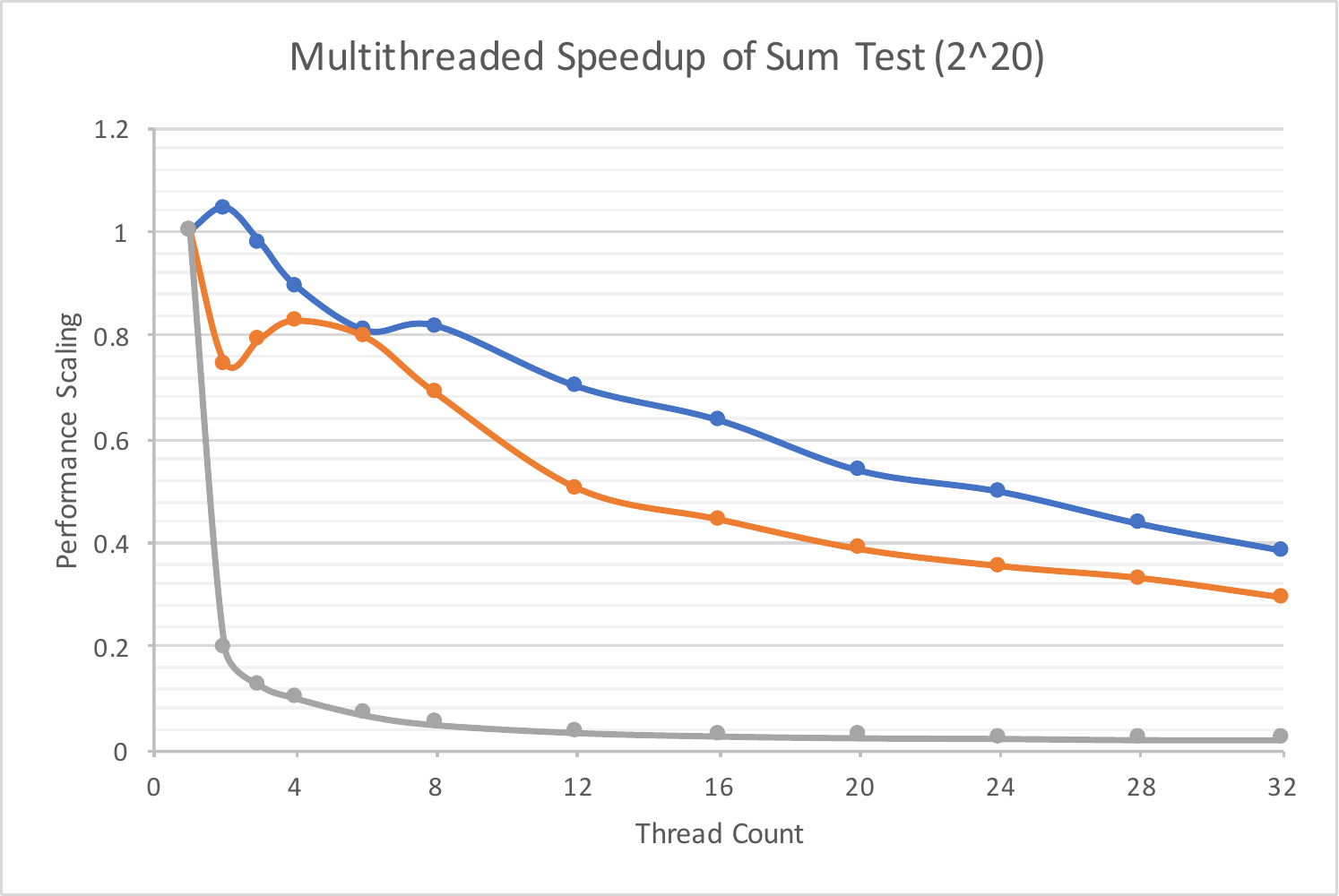}
		\includegraphics[width=0.48\textwidth]{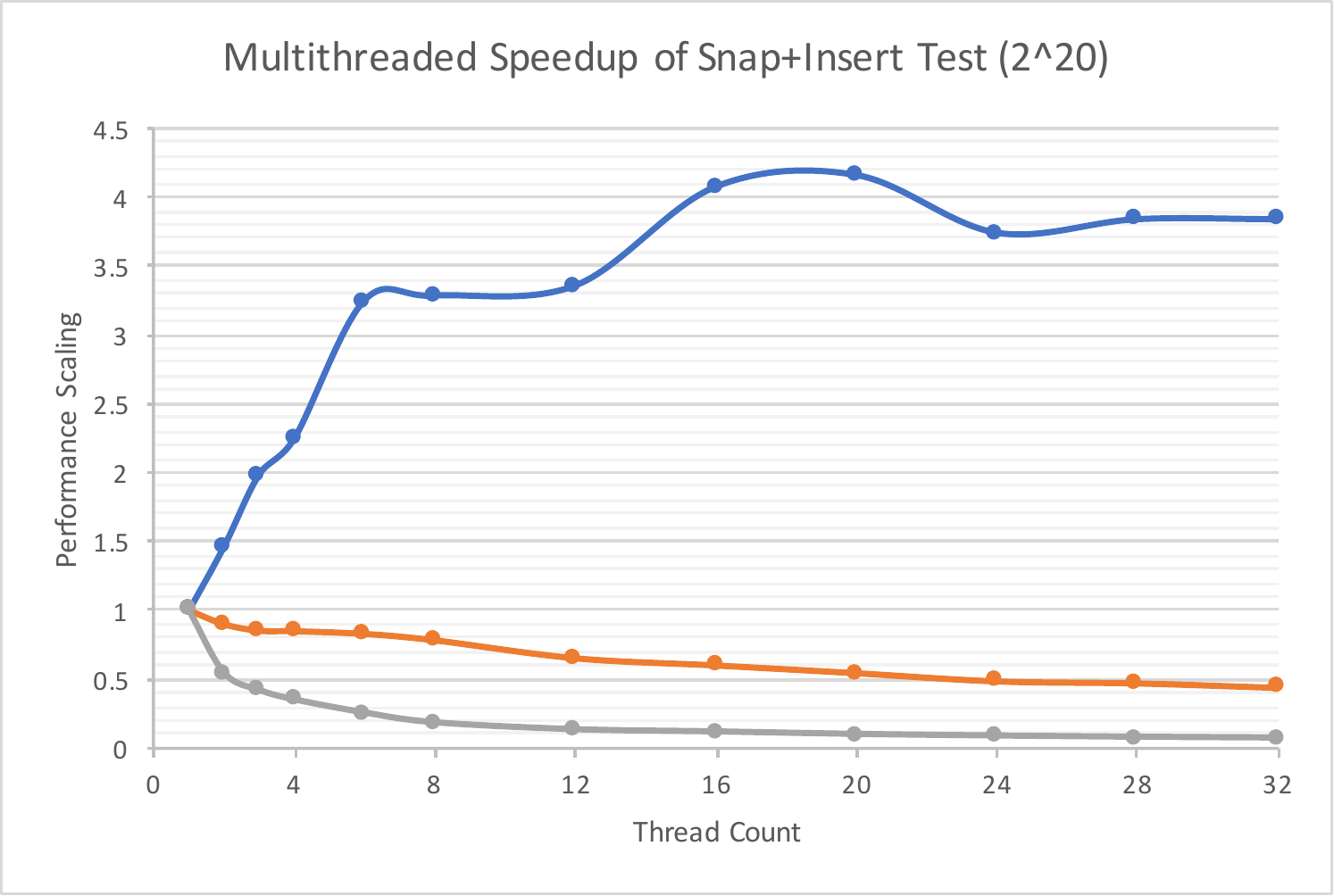}
		\includegraphics[width=0.75\textwidth]{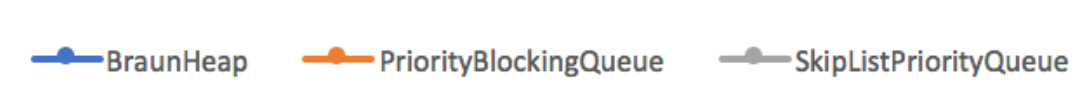}
		\caption{\label{fig:exp20}Left to right, top to bottom, results from the Insert, RemoveMin, Sum, and Snap+Insert tests on priority queues initialized with $2^{20}$ elements.}
	\end{figure*}
	In Figure \ref{fig:exp20}, the results of the first 4 tests, on priority queues of size $\approx2^{20}$, can be seen, plotted as speedups over their single threaded executions.
	Generally, the concurrent Braun heaps scaled better than \texttt{PriorityBlockingQueue}, except on the RemoveMin test, where both algorithms suffered from an O(log n) traversal with a lock on the root, but \texttt{PriorityBlockingQueue} was able to fall back on a highly-tuned single threaded algorithm backed by an array.
	On the tests involving iteration and snapshots, the concurrent Braun heap implementation scaled substantially better than either the \texttt{SkipListPriorityQueue} or the \texttt{PriorityBlockingQueue}.
	Also of note is that Petrank and Timnat's lock-free iterators appear to pay a heavy performance penalty for linearizability in a lock-free setting without the benefit of fast snapshots.
	\begin{figure*}
		\centering
		\includegraphics[width=0.5\textwidth]{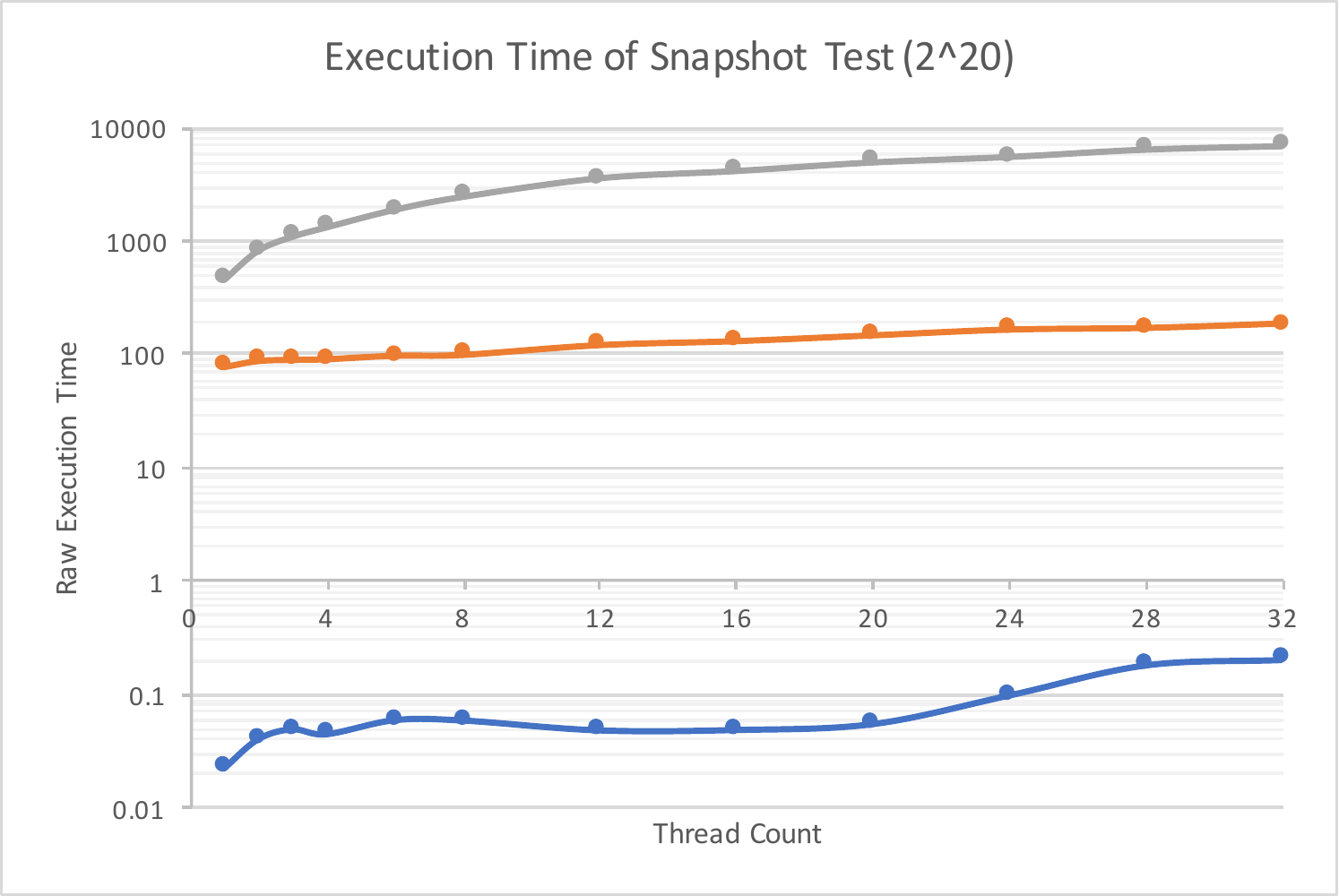}
		\caption{\label{fig:exp20b}Raw execution times of the Snap-Only test, plotted on a log scale.
		The legend coloring matches that of Figure \ref{fig:exp20}.}
	\end{figure*}
	For the four tests shown in Figure \ref{fig:exp20}, the raw performance numbers were within small to moderate constant factors of one another; however, with the Snap-Only test, plotted in Figure \ref{fig:exp20b}, the performance gap between $O(1)$ vs $\Omega(n)$ is fully evident in the execution times, in favor of the concurrent Braun heaps.
	
	\begin{figure*}
		\centering
		\includegraphics[width=0.32\textwidth]{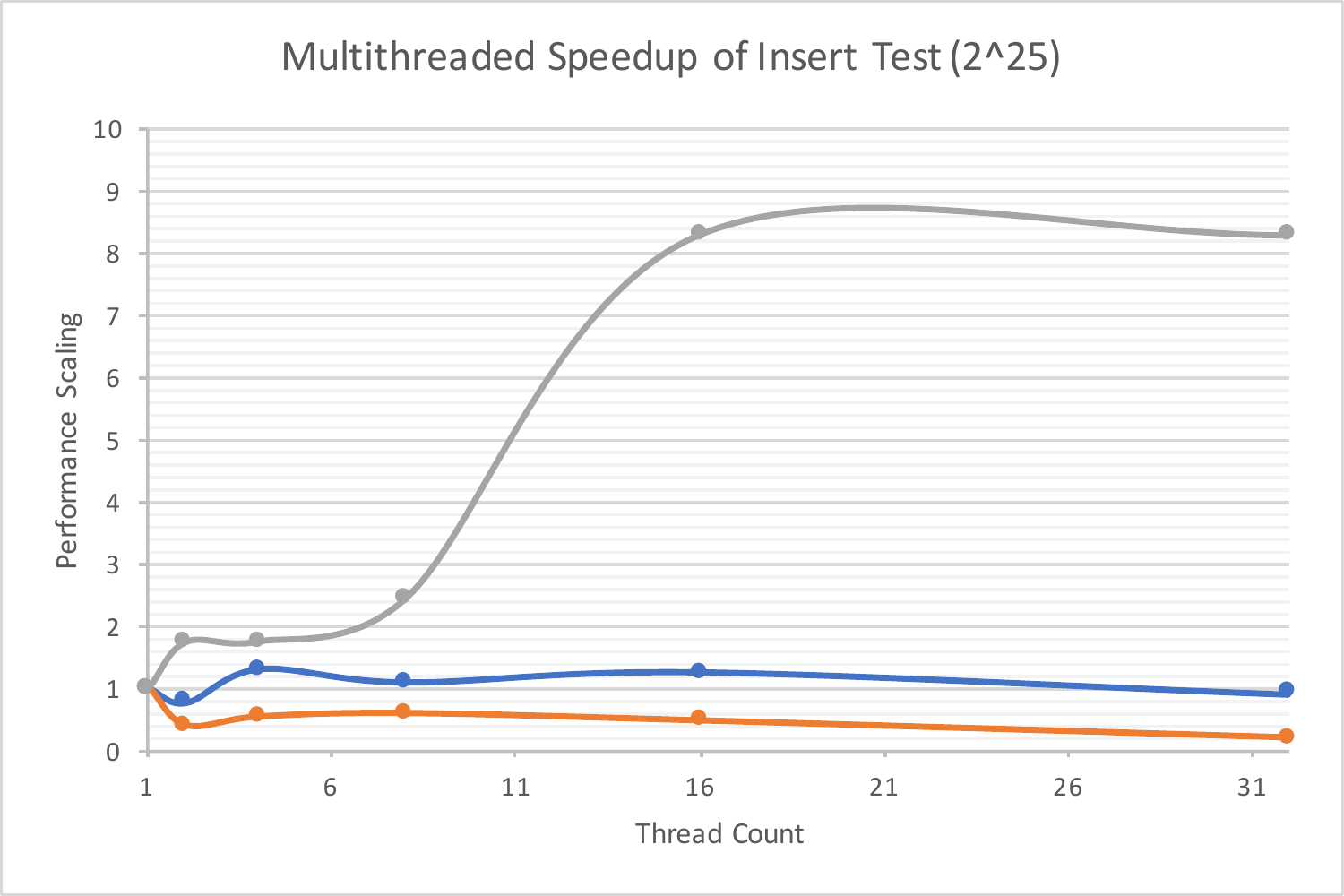}
		\includegraphics[width=0.32\textwidth]{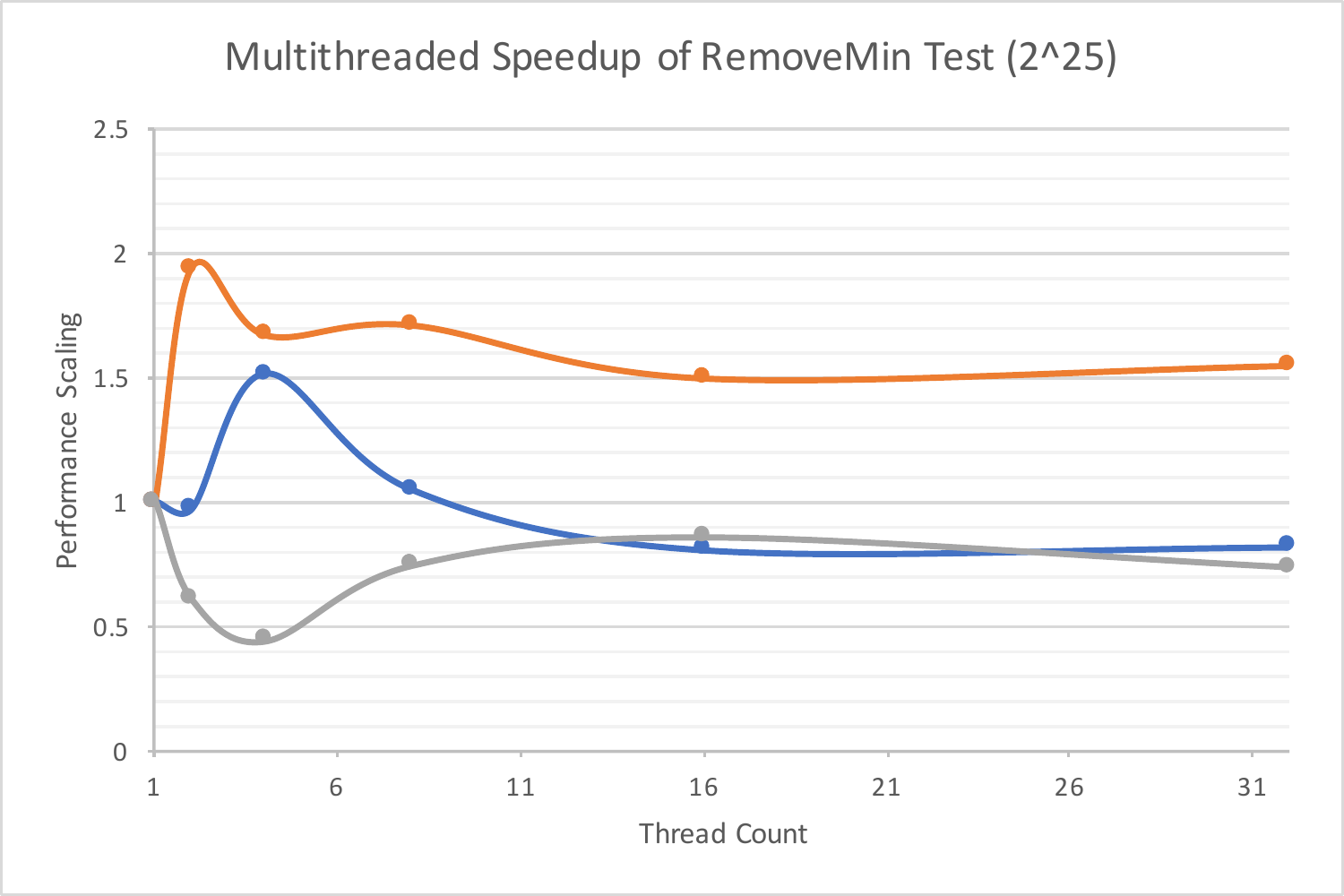}
		\includegraphics[width=0.32\textwidth]{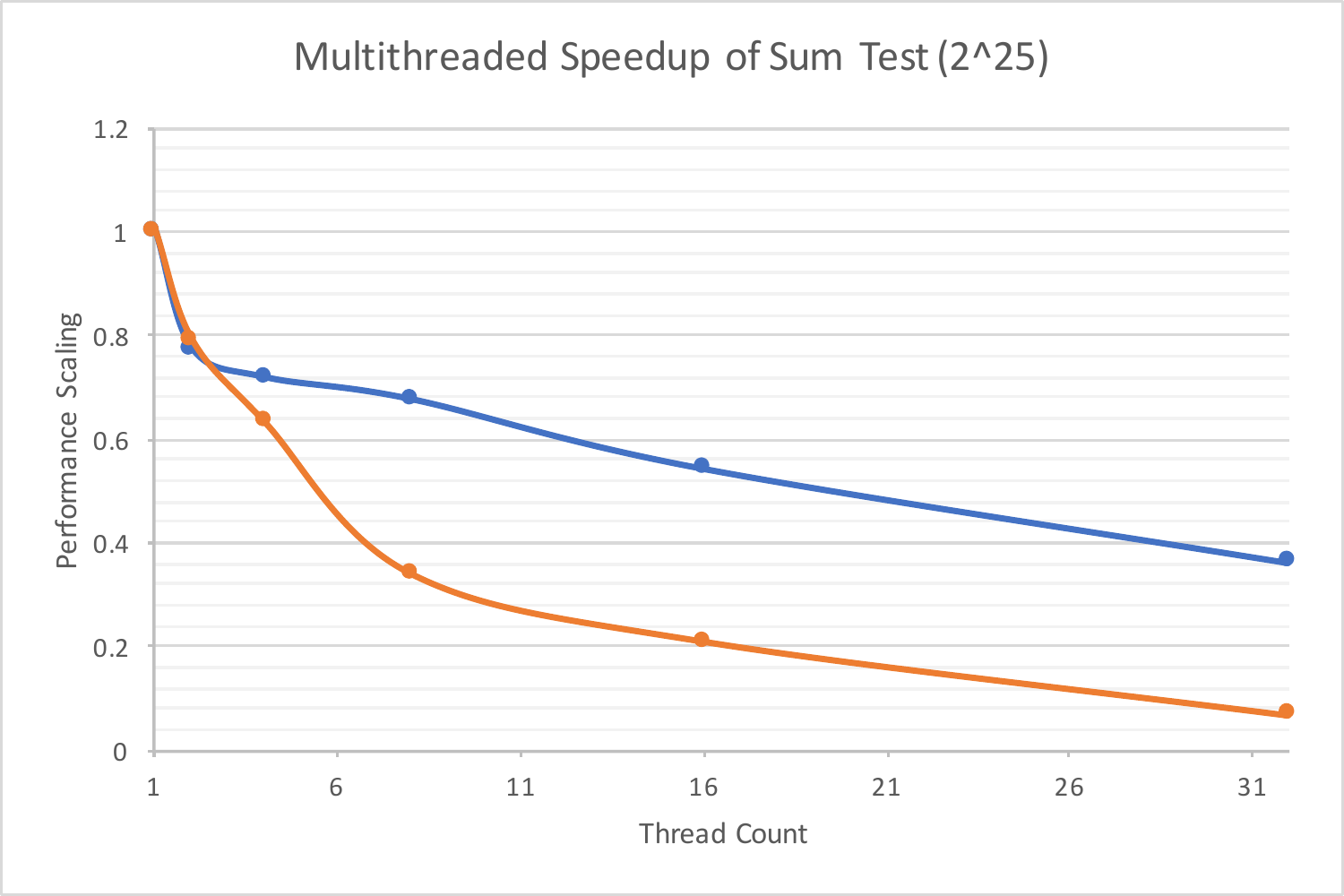}
		\caption{\label{fig:exp25}Left to right, Insert, RemoveMin, and Sum tests on priority queues initialized with $2^{25}$ elements.
		The legend coloring matches that of Figure \ref{fig:exp20}.}
	\end{figure*}
	The Insert, RemoveMin, and Sum tests were repeated on priority queues of size $\approx 2^{25}$.
	The snapshot-based tests were excluded due to the prohibitive memory costs for the SkipList and Blocking implementations\footnote{The SkipList implementation also had to be excluded from the final iterator test due to prohibitive execution times.}.
	For these tests, the performance scaling gaps already on display in earlier tests for Insert and Sum were amplified; however, the results for RemoveMin now favored \texttt{PriorityBlockingQueue}.
	More comprehensive benchmarking is in order to better characterize this regime change.
	
	\section{Conclusion}
	The algorithm for concurrent Braun heaps presented here is both easily implemented\footnote{The \texttt{ConcurrentSkipListMap} source code, on which the competing \texttt{SkipListPriorityQueue} was based, has $>2000$ lines of code, whereas the BraunHeap in the appendix has $<200$.} and easily verified, while providing competitive performance scaling on traditional priority queue operations, and supporting fast and consistent snapshots and iteration.
	
	\subsection{Future Work}
	This work leads naturally to several future lines of inquiry.
	Much of the algorithm lends itself to a lock-free approach based on CAS instructions; however, the two phases of the \texttt{removeMin} operation leave the heap in an inconsistent state if not executed atomically.
	An efficient approach to speculative executions of \texttt{removeMin} would be a useful next step in developing a lock-free concurrent Braun heap (there is also a fairness concern here, as long-running remove operations may consistently get delayed by fast inserts).
	
	The experimental results of this paper also suggest that without fast snapshots, iteration over concurrent data structures is inordinately expensive due to synchronization costs.
	Another direction would be to address this by leveraging copy-on-write semantics in the pursuit of efficient snapshot operations for other data structures, potentially including skip lists.
		
	\bibliography{references}
	\appendix
	\section{Extended source code listings}
	\begin{listing*}
	\begin{minted}{scala}
def throughLock[N,Z](fL:N => Lock)(fZ:N => Z)(n:N):Z = {
    val l = fL(n); l.lock;    val r = fZ(n); l.unlock; r
}
def withNullPath[E:Ordering,Z](v:HeapNode[E], nP: => Z, nnP: HeapNode[E] => Z):Z = {
    if(v == null) nP else nnP(v)
}
def dfsIfNotNull[E:Ordering](n:HeapNode[E]):Stream[E] = {
    withNullPath(n,Stream.Empty,{node => node.busyRead.lock; node.dfs})
}
def snapshotIfNotNull[E:Ordering](n:HeapNode[E]):HeapNode[E] = {
    withNullPath(n, null, _.snapshot)
}
def lockPair:(WriteLock, ReadLock) = {
    val r = new ReentrantReadWriteLock; (r.writeLock, r.readLock)
}
def ensureWriteable[E:Ordering](n:HeapNode[E], mut:HeapNode[E]=>Unit, v:E):HeapNode[E] = {
    val ret = withNullPath(n,new HeapNode(v),{
        concrete:HeapNode[E] =>
            concrete.busyWrite.lock
            concrete.unsnap
    }); mut(ret); ret
}
	\end{minted}
	\caption{\label{lst:parbheaphelpers} Shared helper methods for the concurrent BraunHeap and HeapNode classes.}
	\end{listing*}
	\begin{listing*}
	\begin{minted}{scala}
    @volatile var snapCount = 0
    val (busyWrite, busyRead) = lockPair
    def snapshot:HeapNode[E] = {
        busyWrite.lock; snapCount += 1; busyWrite.unlock; this
    }
    def unsnapHelper:HeapNode[E] = {
        if(snapCount > 0){
            val nV = value 
            val nL = snapshotIfNotNull(left)
            val nR = snapshotIfNotNull(right)
            snapCount -= 1; val ret = new HeapNode(nV, nL, nR)
            ret.busyWrite.lock; busyWrite.unlock; ret
        } else { this }
    }
    def unsnap:HeapNode[E] = {
        throughLock[HeapNode[E],HeapNode[E]](_.busyRead)(_.unsnapHelper)(this)
    }
    private[BraunHeap] def dfs:Stream[E] = {
        busyRead.unlock; ((value #:: dfsIfNotNull(left)) #::: dfsIfNotNull(right))
    }
	\end{minted}
	\caption{\label{lst:parheapnode1} The concurrent HeapNode is augmented with a COW-based snapshot functionality.}
	\end{listing*}
	\begin{listing*}
	\begin{minted}{scala}
class HeapNode[E:Ordering](@volatile var v:E, @volatile var left:HeapNode[E],
                                              @volatile var right:HeapNode[E]){
    def updateHelper(nV:E, nL:HeapNode[E], nR:HeapNode[E]):HeapNode[E] = {
        v = nV; left = nL; right = nR; this
    }
    def update(nV:E, nL:HeapNode[E], nR:HeapNode[E]):HeapNode[E] = {
        unsnap.updateHelper(nV, nL, nR)
    }
    def insertHelper(nV:E):Unit = {
        val (smaller, larger) = if(nV < v){ (nV, v) } else { (v, nV) }
        val insertNeeded = (right != null)
        update(smaller, ensureWriteable(right, right_=, larger),
                        left).insertHelperPost(insertNeeded, larger)
    }
    def insertHelperPost(insertNeeded:Boolean, larger:E):Unit = {
        if(insertNeeded){
            val leftNow = left; busyWrite.unlock; leftNow.insertHelper(larger)            
        } else { busyWrite.unlock }
    }
    def pullupLeftHelper(selfP:HeapNode[E]=>Unit, amRoot:Boolean = false):Option[E] = {
        withNullPath(left,{
            selfP(null); if(snapCount > 0) { snapCount -= 1 } busyWrite.unlock
            if(amRoot){ busyWrite.unlock; None } else { Some(v) }
        },update(v, right, _).pullupLeftHelperPost(selfRef))
    }
    def pullupLeftHelperPost(selfP:HeapNode[E]=>Unit):Option[E] = {
        selfP(this); right.busyWrite.lock; busyWrite.unlock
        right.pullupLeftHelper(right_=)
    }
    def pushdownHelper:Unit = {
        val pushNext = if(right == null && left != null){
            val tmpLV = left.v
            if(tmpLV < v){
                left.busyWrite.lock; left = left.update(v)
                left.busyWrite.unlock; value = tmpLV
            } null
        } else if(right != null) {
            val (tmpV, tmpLV, tmpRV) = (v, left.value, right.value)
            if((tmpLV < tmpV) || (tmpRV < tmpV)){
                val (nV, nU, nP) = if(tmpLV < tmpRV) {
                    (tmpLV, left_= _, left)
                } else { (tmpRV, right_= _, right) }
                v = nV; nP.busyWrite.lock; val ret = nP.update(tmpV); nU(ret); ret
            } else { null }
        } else { null }
        busyWrite.unlock
        if(pushNext != null) pushNext.pushdownHelper
    }
}
	\end{minted}
	\caption{\label{lst:parheapnode2} The core BraunHeap algorithm is unchanged by hand-over-hand locking. Some methods have been refactored to respect COW boundaries.}
	\end{listing*}
	\begin{listing*}
	\begin{minted}{scala}
class BraunHeap[E:Ordering] {
    @volatile var root:HeapNode[E] = initRoot
    val (busyWrite, busyRead) = lockPair
    def syncReadRoot:HeapNode[E] = throughLock(_.busyRead)(_.root)(this)
    def snapshot:BraunHeap[E] = {
        new BraunHeap[E](snapshotIfNotNull(syncReadRoot))
    }
    def iterator:Iterator[E] = dfsIfNotNull(snapshot.root).iterator  
    def getMin:Option[E] = {
        Option(syncReadRoot).map{throughLock(_.busyRead)(_.value) _}
    }
    def insertPre(value:E):Boolean = {
        val insertNeeded = (root != null)
        ensureWriteable(root, root_=, value)
        insertNeeded
    }
    def insert(value:E):Unit =  {
        if(throughLock[BraunHeap[E],Boolean](_.busyWrite){
            _.insertPre(value)
        }(this)){ root.insertHelper(value) }
    }
    def removeMinPre:(Option[E], Option[E]) = {
        withNullPath(root,(None, None),{
                rootInit:HeapNode[E] =>
                    rootInit.busyWrite.lock
                    root = rootInit.unsnap
                    root.busyWrite.lock
                    val oldRoot = Some(root.value);
                    val newRoot = root.pullupLeftHelper(root_=, true)
                    (oldRoot, newRoot)
            })
    }
    def removeMin:Option[E] = {
        val (ret, newRoot) = throughLock(_.busyWrite){_.removeMinPre}(this)
        if(ret.isDefined && newRoot.isDefined){
            root.value = newRoot.get
            root.pushdownHelper
        } ret
    }
}
	\end{minted}
	\caption{\label{lst:parbraunheap} The core BraunHeap algorithm is unchanged by hand-over-hand locking.}
	\end{listing*}
\end{document}